\documentclass{article}
\usepackage{ijcai18}

\usepackage{times}
\usepackage{xcolor}
\usepackage{soul}
\usepackage[utf8]{inputenc}
\usepackage[small]{caption}

\usepackage{amsthm}
\usepackage{amssymb}
\usepackage{listings}
\usepackage{multicol}
\usepackage{verbatim}
\usepackage{syntax}
\usepackage{enumitem}

\theoremstyle{definition}
\newtheorem{definition}{Definition}
\newtheorem{example}{Example}
\newtheorem{theorem}{Theorem}
\newtheorem{lemma}{Lemma}
\newtheorem{proposition}{Proposition}
\newtheorem{corollary}{Corollary}
\newtheorem*{priorwork*}{Prior Work}

\title{On the Conditional Logic of Simulation Models\footnote{In \emph{Proceedings of the 27th International Joint Conference on Artificial Intelligence} (IJCAI 2018).}}

\author{
Duligur Ibeling$^1$, 
Thomas Icard$^2$,
\\
$^1$ Department of Computer Science, Stanford University \\
$^2$ Department of Philosophy, Stanford University \\
duligur@stanford.edu,
icard@stanford.edu
}
\date{}
\let\ACMmaketitle=\maketitle
\renewcommand{\maketitle}{\begingroup\let\footnote=\thanks \ACMmaketitle\endgroup}

\begin{document}

\maketitle

\begin{abstract}
We propose analyzing conditional reasoning by appeal to a notion of intervention on a simulation program,
formalizing and subsuming a number of approaches to conditional thinking in the recent AI literature.
Our main results include a series of axiomatizations, allowing comparison between this framework and existing frameworks (normality-ordering models, causal structural equation models), and a complexity result establishing $\mathsf{NP}$-completeness of the satisfiability problem.
Perhaps surprisingly, some of the basic logical principles common to all existing approaches are invalidated in our causal simulation approach. We suggest that this additional flexibility is important in modeling some intuitive examples.
\end{abstract}


\section{Introduction and Motivation}
Much of intelligent action and reasoning involves assessing what \emph{would} occur (or \emph{would have} occurred)
under various non-actual conditions. Such hypothetical and counterfactual (broadly, \emph{subjunctive}) conditionals are bound up with central topics in artificial intelligence, including prediction, explanation, causal reasoning, and decision making. It is thus for good reason that AI researchers have focused a great deal of attention on conditional reasoning (see, e.g., \cite{Ginsberg,Delgrande,FriedmanKollerHalpern,Pearl2009,Bottou}, among many others).

Two broad approaches to subjunctive conditionals have been especially salient in the literature. The first, originating in philosophy \cite{Stalnaker,lewis73}, takes as basic a ``similarity'' or ``normality'' ordering on possibilities, and evaluates a claim `if $\varphi$ then $\psi$' by asking whether $\psi$ is true in (e.g., all)
the most normal $\varphi$ possibilities. The second approach, associated with the work of Judea Pearl, takes as basic a causal
``structural equation'' model (SEM), and evaluates conditionals according to a defined notion of \emph{intervention} on the
model. These two approaches are in some technical and conceptual respects compatible \cite{Pearl2009}, though they can also be shown to conflict on some basic logical matters \cite{Halpern2013}. Both capture important intuitions about conditional reasoning, and both have enjoyed successful applications in AI research.

In this article we propose a third approach to conditionals, which captures a different intuition, and which can already be seen as implicit in a growing body of work in AI, as well as in cognitive science. This approach takes as basic the notion of a \emph{simulation model}, that is, a \emph{program} for simulating the transformation from one state of the world to another, or for \emph{building up} or  
\emph{generating} a world from a partial description of it. Simulation models have been of interest since the earliest days of AI \cite{NewellSimon}. A recent tradition, coming out of work on statistical relational models, has proposed building complex generative models using rich and expressive programming languages, typically also incorporating probability (e.g., \cite{Pfeffer,BLOG,goodmanetal08,deRaedt}). Such languages have also been used for modeling human reasoning, including with counterfactuals \cite{Goodman2014}.

Simulation models have an obvious causal (and more general \emph{dependence}) structure, and it is natural to link conditionals with this very structure. We can assess a claim `if $\varphi$ then $\psi$' by \emph{intervening} on the program to ensure that $\varphi$ holds true throughout the simulation, and asking whether $\psi$ holds upon termination. This is conceptually different from the role of intervention in structural equation models, where the post-intervention operation is to find solutions to the manipulated system of equations. As we shall see, this conceptual difference has fundamental logical ramifications. 

This more procedural way of thinking about subjunctive conditionals enjoys various advantages. First, there is empirical evidence suggesting that human causal and conditional reasoning is closely tied to mental simulation \cite{sloman05}. Second, there are many independent reasons to build generative models in AI (e.g., minimizing prediction error in classification; see \cite{LiangJordan}), making them a common tool. Thus, opportunistically, we can expect to have such models readily available (perhaps unlike normality orderings or even structural equation models).

Related to this second point, many of the generative models that are currently being built using deep neural networks fit neatly into our approach, even though we can often only use them as black boxes (see, e.g., \cite{cgan,CausalGAN}, etc.). We know how to intervene on these programs (i.e., controlling input), and how to read off a result or prediction---that is, we can observe what conditional claims they embody---even though we may not understand all the causal details of the learned model. Some authors have recently argued that certain kinds of counterfactual analysis in particular establish an appropriate standard for \emph{interpretability} for these models \cite{wachter}.

Our contribution in this article is threefold: (1) we propose a general semantic analysis of conditional claims in terms of program executions, subsuming all the aforementioned application areas; (2) we establish completeness theorems for a propositional conditional language with respect to (four different classes of) programs, allowing a comparison with alternative approaches at a fundamental logical level; (3) we establish $\mathsf{NP}$-completeness of the satisfiability problem for these logical systems. Before turning to these details, we explain informally what is distinctive about the resulting logic. 

\section{Conditional Logics} \label{logics} 

The literature on conditional logic is extensive. We focus here on the most notable differences between the systems below and more familiar systems based on either world-orderings or SEMs. We will be using a notation inspired by dynamic logic (also used by \cite{Halpern2000}), whereby $[\alpha]\beta$ can loosely be read as, `if $\alpha$ were true, then $\beta$ would be true.' Understanding the complete logic of a given interpretation can be of both theoretical and practical interest. In the causal setting, for instance, a complete set of axioms may give the exact conditions under which some counterfactual quantity is (not) identifiable from statistical data \cite{Pearl2009}.  

One of the bedrock principles of conditional reasoning is called \emph{Cautious Monotonicity} \cite{Kraus1990}, or sometimes the \emph{Composition} rule \cite{Pearl2009}. This says that from $[A](B \wedge C)$ we may always infer $[A \wedge B]C$. While there are known counterexamples to it in the literature---it fails for some probabilistic and possibilistic interpretations \cite{DuboisPrade} and in standard versions of default logic \cite{Makinson}---the principle is foundational to both world-ordering models and SEMs.
By contrast, in our setting, holding $B$ fixed during the simulation may interrupt the sequence of steps leading to $C$ being made true. Here is a simple example (taken from \cite{AC17}):
\begin{example} If Alf were ever in trouble ($A$), the neighbors Bea and Cam would both like to help ($B$ and $C$, respectively). But neither wants to help if the other is already helping. Imagine the following scenario: upon finding out that Alf is in trouble, each looks to see if the other is already there to help. If not, then each begins to prepare to help, eventually making their way to Alf but never stopping again to see if the other is doing the same. If instead, e.g., Cam initially sees Bea already going to help, Cam will not go. One might then argue that the following both truly describe the situation: `If Alf were in trouble, Bea and Cam would both go to help' and
`If Alf were in trouble and Bea were going to help, Cam would not go to help'.
\label{ex1}
\end{example}
The example trades on a temporal ambiguity about when Bea is going to help, and it can be blocked  simply by time-indexing variables. However, following a common stance in the literature \cite{Halpern2000,Pearl2009}, we maintain that  requiring temporal information always be made explicit is excessively stringent. Furthermore, in line with our earlier remarks about black box models, we may often be in a situation where we simply do not understand the internal temporal and causal structure of the program. To take a simple example, asking a generative image model to produce a cityscape might result in images with clouds and blue skies, even though a request to produce a cityscape with a blue sky might not result in any clouds. We would like a framework that can accommodate conditional theories embodied in artifacts like these.

Our completeness results below (Thm. \ref{ax-soundcomplete}) show that the logic of conditional simulation is strictly weaker than any logic of structural equation models (as established in \cite{Halpern2000}) or of normality orderings (as, e.g., in \cite{lewis73}). The conditional logic of \emph{all} programs is very weak indeed. At the same time, some of the axioms in these frameworks can be recovered by restricting the class of programs (e.g., the principle of \emph{Conditional Excluded Middle}, valid on structural equation models and on some world-ordering models \cite{Stalnaker}, follows from optional axiom $\textsf{F}$ below). We view this additional flexibility as a feature. However, even for a reader who is not convinced of this, we submit that understanding the logic of this increasingly popular way of thinking about conditional information is valuable.


\begin{priorwork*}
The notion of intervention introduced below (Defn. \ref{intervention}) is different from, but inspired by, the corresponding notion in SEMs \cite{meekglymour94,Pearl2009}. The logical language we study in this paper, restricting antecedents to conjunctive clauses but closing off under Boolean connectives, follows \cite{Halpern2000}. 

Interestingly, prior to any of this work, \cite{Balkenius} studied conditionals interpreted specifically over certain classes of neural networks, using a definition of ``clamping a node'' similar to our notion of intervention. They also observed that some of the core principles of non-monotonic logic fail for that setting. (See in addition \cite{Leitgeb} for further development of related ideas.)
\end{priorwork*}

\section{Syntax}
Let $X$ be a set of atoms $X_1, X_2, X_3, \dots$ and
let $\mathcal{L}_{\mathrm{prop}}$ be the language of propositional formulas over atoms in $X$
closed under disjunction, conjunction, and negation. Let $\mathcal{L}_{\mathrm{int}} \subset \mathcal{L}_{\mathrm{prop}}$
be the language of
purely conjunctive, ordered formulas of unique literals, i.e., formulas of the form $l_{i_1} \land \dots \land l_{i_n}$, where $i_j < i_{j+1}$ and each $l_{i_j}$ is either $X_{i_j}$ or $\lnot X_{i_j}$.
Each formula in $\mathcal{L}_{\mathrm{int}}$ will specify an \emph{intervention} by giving fixed values for a fixed list of variables.
We also include the ``empty'' intervention $\top$ in $\mathcal{L}_{\mathrm{int}}$. Given $\varphi \in \mathcal{L}_{\mathrm{prop}}$,
$\varphi' \in \mathcal{L}_{\mathrm{int}}$ is the \emph{$\mathcal{L}_{\mathrm{int}}$-equivalent}
of $\varphi$
if $\varphi$ is a propositionally consistent, purely conjunctive formula over literals and $\varphi'$ results from a reordering
of literals and deletion of repeated literals
in $\varphi$. For example, the $\mathcal{L}_{\mathrm{int}}$-equivalent of $\lnot X_2 \land X_1 \land X_1$ is $X_1 \land \lnot X_2$.
Let $\mathcal{L}_{\mathrm{cond}}$ be the language of formulas
of the form $[\alpha] \beta$ for $\alpha \in \mathcal{L}_{\mathrm{int}}, \beta \in \mathcal{L}_{\mathrm{prop}}$. We call such a formula a
\emph{subjunctive conditional}, and call $\alpha$
the \emph{antecedent} and $\beta$ the \emph{consequent}.
The overall causal simulation language $\mathcal{L}$ is the language of propositional
formulas over atoms in $X \cup \mathcal{L}_{\mathrm{cond}}$
closed under disjunction, conjunction, and negation.
For $\alpha, \beta \in \mathcal{L}$,
$\alpha \rightarrow \beta$ abbreviates
$\lnot \alpha \lor \beta$,
and $\alpha \leftrightarrow \beta$ denotes
$(\alpha \rightarrow \beta) \land (\beta \rightarrow \alpha)$.
We use $\langle \alpha \rangle$ for the dual of $[\alpha]$,
i.e., $\langle \alpha \rangle \beta$ abbreviates $\lnot [\alpha] (\lnot \beta)$.

\section{Semantics}
We now define the semantics of $\mathcal{L}$ over causal simulation models.
A \emph{causal simulation model} is
a pair $(\mathsf{T}, \mathbf{x})$ of a Turing machine $\mathsf{T}$
and tape contents represented by a \emph{state description}
$\mathbf{x} = \{x_n\}_{n \in \mathbb{N}}$, which
specifies binary\footnote{The present setting can be easily generalized to the arbitrary discrete setting, indeed without changing the logic. See \cite{AC17}.} values for all tape variables,
only finitely many of which can be nonzero.
Running $\mathsf{T}$ on input $\mathbf{x}$ yields a new
state description $\mathbf{x}'$ as output, provided the execution halts.
We say $\mathbf{x} \models X_i$ iff $x_i = 1$ in $\mathbf{x}$. Satisfaction
$\mathbf{x} \models \varphi$ of $\varphi \in \mathcal{L}_{\mathrm{prop}}$
is then defined in the familiar way by recursion.
For $X$-atoms
we define
$(\mathsf{T}, \mathbf{x}) \models X_i$ iff $\mathbf{x} \models X_i$.
Toward a definition of satisfaction for subjunctive conditionals,
we now define an \emph{intervention} (in the same way as in \cite{AC17}):
\begin{definition}[Intervention]
An intervention $\mathcal{I}$ is a computable function mapping a machine $\mathsf{T}$ to a new
machine $\mathcal{I}(\mathsf{T})$ by taking a set of values $\{x_i\}_{i \in I}$,
$I \subseteq \mathbb{N}$ a finite index set,
and holding fixed the value of each $X_i$ to $x_i$ throughout the execution of $\mathsf{T}$.
That is, $\mathcal{I}(\mathsf{T})$ first sets each $X_i$ to $x_i$, then runs $\mathsf{T}$
while ignoring any write to any $X_i$.
\label{intervention}
\end{definition}
Any $\alpha \in \mathcal{L}_{\mathrm{int}}$ uniquely specifies an intervention, which
we denote as $\mathcal{I}_{\alpha}$:
each literal in $\alpha$ gives a tape variable to hold fixed, and the literal's polarity tells us to which value it is to be fixed.
Now we define $(\mathsf{T}, \mathbf{x}) \models [\alpha]\beta$
iff for all halting executions of $\mathcal{I}_\alpha(\mathsf{T})$
on $\mathbf{x}$, the resulting tape
satisfies $\beta$.
Note that for deterministic machines, this means either $\mathcal{I}_\alpha(\mathsf{T})$
does not halt on $\mathbf{x}$, or the unique resulting tape satisfies $\beta$.
The definition also implies that $(\mathsf{T}, \mathbf{x}) \models \langle \alpha \rangle \beta$
iff {there exists} a halting execution of $\mathcal{I}_\alpha(\mathsf{T})$
on $\mathbf{x}$ whose result satisfies $\beta$.
Having now defined $(\mathsf{T}, \mathbf{x}) \models \varphi$ for atoms $\varphi \in X \cup \mathcal{L}_{\mathrm{cond}}$,
$(\mathsf{T}, \mathbf{x}) \models \varphi$
for complex $\varphi \in \mathcal{L}$ is defined by recursion.

Interestingly, as revealed by Prop. \ref{model-checking}, model checking in this setting is difficult, while satisfiability (or validity) for notable classes of machines is decidable (Thm. \ref{complexity}).

\begin{proposition} If $\alpha \wedge \beta$ is propositionally consistent,
then it is undecidable whether $(\mathsf{T},\mathbf{x}) \models \langle \alpha \rangle \beta$.
\label{model-checking}
\end{proposition}
\begin{proof}[Proof Sketch]
Under a suitable encoding of natural numbers on the variable tape, the class $\mathcal{T}_{\alpha} = \{\mathcal{I}_{\alpha}(\mathsf{T}): \mathsf{T} \in \mathcal{T}\}$, where $\mathcal{T}$ is the class of all machines,
gives an enumerable list of all the partial recursive functions, with $\mathsf{T}$ computably recoverable from $\mathsf{T}' \in \mathcal{T}_{\alpha}$. Moreover, $H_{\beta} = \{\mathsf{T} \in \mathcal{T}_{\alpha}: \mathsf{T}$ halts on input $\mathbf{x}$ with output $\mathbf{x}' \models \beta\}$ is extensional and $\varnothing \subsetneq H_{\beta} \subsetneq \mathcal{T}_{\alpha}$, so by the Rice-Myhill-Shapiro Theorem it is undecidable. If we could decide whether $(\mathsf{T},\mathbf{x}) \models \langle \alpha \rangle \beta$, this would allow us to decide whether $\mathsf{T}' = \mathcal{I}_{\alpha}(\mathsf{T}) \in H_{\beta}$.
\end{proof}

A second limitative result is that we cannot have strong completeness (that is, completeness relative to arbitrary sets of assumptions), since by Prop. \ref{noncompact} we do not have compactness. On the other hand, our axiom systems (Defn. \ref{systems}) are weakly complete (complete relative to finite assumption sets).
\begin{proposition} The language $\mathcal{L}$ interpreted over causal simulation models is not compact.
\label{noncompact}
\end{proposition} \begin{proof} Let $f:\mathbb{N}\rightarrow\mathbb{N}$ be any uncomputable total function such that $f(n) \neq n$ for all $n$ and consider $\Omega = \{\neg X_n : n \in \mathbb{N}\} \cup \{\langle X_n \rangle X_{f(n)}: n \in \mathbb{N} \} \cup \{[X_n]\neg X_m: m, n \in \mathbb{N}$ with $m \neq n, m \neq f(n)\}$. If $(\mathsf{T},\mathbf{x})$ satisfies every $\varphi \in \Omega$, we could compute $f(n)$ by intervening to set $X_n$ to 1, and checking which other variable $X_m$ is set to 1. As $f$ is total and $f(n)\neq n$, we could always find such $m = f(n)$. So $\Omega$ is unsatisfiable.  But it is easily seen that every finite subset of $\Omega$ is satisfiable. 
\end{proof}

\section{Axiomatic Systems}
We will now identify axiomatic systems (Defn. \ref{systems}) that are \emph{sound} and \emph{complete} with respect to salient 
classes (Defn. \ref{classes}) of
causal simulation models, by which we mean that they prove all (completeness) and only (soundness) the \emph{generally valid principles} with respect to
those classes.

\begin{definition} Let $\mathcal{M}$ be the class of all causal simulation models $(\mathsf{T}, \mathbf{x})$, where $\mathsf{T}$ may be non-deterministic. Let $\mathcal{M}_{\mathrm{det}}$ be the class of models with deterministic $\mathsf{T}$, and let $\mathcal{M}^{\downarrow}$ be the class of models with non-deterministic $\mathsf{T}$ that halt on all input tapes and interventions. 
Also let $\mathcal{M}^{\downarrow}_{\mathrm{det}} = \mathcal{M}_{\mathrm{det}} \cap
\mathcal{M}^{\downarrow}$.
\label{classes}
\end{definition} 

\begin{definition} Below are two rules and four axioms.\footnote{
We use the standard names from modal and non-monotonic logic.
The \emph{Left Equivalence} rule \cite{Kraus1990},
namely, infer $[\alpha] \beta \leftrightarrow [\alpha'] \beta$ from $\alpha \leftrightarrow \alpha'$, is not needed: since antecedents belong to $\mathcal{L}_{\mathrm{int}}$,
they are never distinguished beyond equivalence.
} \begin{eqnarray*}
 \textsf{PC}. && \mbox{Propositional calculus (over the atoms of $\mathcal{L}$)} \\
 \textsf{RW}.& & \mbox{From }\beta \rightarrow \beta'\mbox{ infer }[\alpha]\beta \rightarrow [\alpha]\beta' \\
  \textsf{R}. && [\alpha]\alpha \\
 \textsf{K}. && [\alpha](\beta \rightarrow \gamma) \rightarrow ([\alpha]\beta \rightarrow [\alpha]\gamma) \\ 
 \textsf{F}. && \langle \alpha\rangle \beta \rightarrow [\alpha]\beta \\
 \textsf{D}. && [\alpha]\beta \rightarrow \langle \alpha \rangle \beta 
\end{eqnarray*}
$\textsf{AX}$ denotes the system containing axioms \textsf{R} and \textsf{K} and
closed under \textsf{PC} and \textsf{RW}.
$\textsf{AX}_{\mathrm{det}}$ is $\textsf{AX}$ in addition to axiom
\textsf{F}, $\textsf{AX}^\downarrow$ is $\textsf{AX}$ in addition
to axiom \textsf{D}, and $\textsf{AX}^{\downarrow}_{\mathrm{det}}$ is the system combining
all of these axioms and rules.
\label{systems}
\end{definition}

For the remainder of this article,
fix $\mathcal{M}^{\dagger}$ to be one of the classes $\mathcal{M}$, $\mathcal{M}_{\mathrm{det}}$, $\mathcal{M}^{\downarrow}$, or $\mathcal{M}_{\mathrm{det}}^{\downarrow}$, and let $\textsf{AX}^{\dagger}$ be the respective deductive system of Defn. \ref{systems}. Then:
\begin{theorem}
$\textsf{AX}^{\dagger}$ is sound and complete
for validities
with respect to the class $\mathcal{M}^{\dagger}$.
\label{ax-soundcomplete}
\end{theorem}
\begin{proof}
The soundness of $\textsf{PC}$,
$\textsf{RW}$, $\textsf{R}$, and $\textsf{K}$
is straightforward.
If $\mathcal{M}^{\dagger}$ is $\mathcal{M}_{\mathrm{det}}$ (or $\mathcal{M}_{\mathrm{det}}^{\downarrow}$),
any $\mathfrak{M} \in \mathcal{M}^{\dagger}$
has at most one halting execution, so a property
holding of one execution holds of all and $\textsf{F}$ is sound.
If $\mathcal{M}^{\dagger}$ is $\mathcal{M}^{\downarrow}$ (or $\mathcal{M}_{\mathrm{det}}^{\downarrow}$),
then any $\mathfrak{M}$ has at least one halting execution,
so a property holding of all holds of one, and
$\textsf{D}$ is sound.

As for completeness,
it suffices to
show that any $\textsf{AX}^{\dagger}$-consistent $\varphi$ has a 
canonical
model $\mathfrak{M}_{\varphi} \in \mathcal{M}^{\dagger}$ satisfying it.
Working toward the construction of $\mathfrak{M}_{\varphi}$,
we prove a normal form result (Lem. \ref{normal-form})
that elucidates what is required in order to satisfy $\varphi$ (Lem. \ref{norm-sf-sat}).
We then define simple
programming
languages (Defn. \ref{pl})---easily seen to be translatable into
Turing machine code---that we employ to construct a program for
$\mathfrak{M}_{\varphi}$ that meets exactly these requirements.


\begin{lemma} Any $\varphi \in \mathcal{L}$ is provably-in-$\textsf{AX}$ (and -$\textsf{AX}^{\dagger}$)
 equivalent to a disjunction of conjunctive clauses, where each clause is of the form
\begin{equation} \pi \wedge \bigwedge_{i \in I} \big([\alpha_i] \bigvee_{j \in J_i} \beta_{j}\big) \wedge \bigwedge_{k \in K} \langle \alpha_k \rangle \beta_k \label{normal}
\end{equation}
and $\pi \in \mathcal{L}_{\mathrm{prop}}$ while $\beta_{j}, \beta_k \in \mathcal{L}_{\mathrm{int}}$ for all $j \in J_i$ for all $i \in I$ and for all $k \in K$.
We may assume without loss of generality that $\alpha_i \neq \alpha_{i'}$ for distinct $i,i' \in I$. \label{normal-form}
\end{lemma}
\begin{proof}
Note that provably in $\textsf{AX}$, $[\alpha] (\beta \land \gamma) \leftrightarrow [\alpha] \beta \land [\alpha] \gamma$ and $\langle \alpha \rangle (\beta \lor \gamma) \leftrightarrow \langle \alpha \rangle \beta \lor
\langle \alpha \rangle \gamma$. 
Use these equivalences and $\textsf{PC}$ and $\textsf{RW}$ to rewrite and get the result.
\end{proof}

Given a clause $\delta$ as in (\ref{normal}), let $\mathcal{S}_\delta \subset \mathcal{L}_{\mathrm{int}}$
be the set of $\mathcal{L}_{\mathrm{cond}}$-antecedents appearing in
$\delta$.
Each $\delta$ gives rise to a \emph{selection function} $f_\delta: \mathcal{S}_\delta \rightarrow \wp(\mathcal{L}_{\mathrm{int}})$ (cf. \cite{Stalnaker}), obtained (not uniquely) as follows. To give the value of $f_\delta(\alpha)$, suppose that $\alpha=\alpha_k$ for some $k\in K$. If $\alpha = \alpha_i$ for some $i \in I$, then $\alpha \wedge \beta_k \wedge \bigvee_{j \in J_i} \beta_j$ is consistent: otherwise, $[\alpha]\bigvee_{j \in J_i} \beta_j \land \langle \alpha \rangle \beta_k$
implies $\langle \alpha \rangle \bot$ which is $\mathsf{AX}$- (and $\mathsf{AX}^{\dagger}$) inconsistent.
Thus for some $j \in J_i$, $\alpha \wedge \beta_k \wedge \beta_j$ is also consistent.
In general $\alpha$ may be $\alpha_k$ for multiple $k \in K$. For each such $k$, we find such a $\beta_j$. We then set $f_\delta(\alpha)$ to the set of $\mathcal{L}_{\mathrm{int}}$-equivalents of the
$\alpha \wedge \beta_k \wedge \beta_j$, and set $f_\delta(\alpha)$ to the set of $\mathcal{L}_{\mathrm{int}}$-equivalents
of the $\alpha \wedge \beta_k$, if $\alpha \neq \alpha_i$ for any $i \in I$.
The remaining case is that $\alpha \in \mathcal{S}_\delta$ but $\alpha \neq \alpha_k$ for any $k\in K$;
in this case, set $f_{\delta}(\alpha) = \varnothing$.
 
\begin{lemma} If $\textsf{AX}^{\dagger}$ is $\textsf{AX}_{\mathrm{det}}$ or $\textsf{AX}^{\downarrow}_{\mathrm{det}}$ we can assume $f_{\delta}(\alpha)$ is a singleton (or possibly empty in the case of $\textsf{AX}_{\mathrm{det}}$). If $\textsf{AX}^{\dagger}$ is $\textsf{AX}^\downarrow$  or $\textsf{AX}^{\downarrow}_{\mathrm{det}}$ we can assume that $\varnothing \notin \mbox{range}(f_{\delta})$.
\label{sf-prop}
\end{lemma}
\begin{proof} In $\textsf{AX}_{\mathrm{det}}$, if $\langle \alpha \rangle \beta_1$ and $\langle \alpha \rangle \beta_2$, then because $[\alpha]\beta_1$ and $[\alpha]\beta_2$, and thus $[\alpha](\beta_1 \wedge \beta_2)$, we have $\langle \alpha \rangle (\beta_1 \wedge \beta_2)$.
In $\textsf{AX}^\downarrow$ it is always possible to assume that for each $i \in I$ there is some $j \in J_i$ such that $\langle \alpha_i \rangle \beta_j$ appears as a conjunct. So no such $\alpha_i$ will be sent to $\varnothing$.
\end{proof}
\begin{lemma} Let $\delta$ be a disjunct as in (\ref{normal}). Let $\mathfrak{M} \in \mathcal{M}$. Suppose that $\mathfrak{M} \models \pi$, and for all $\alpha \in \mathcal{S}_{\delta}$ that
$\mathfrak{M} \models \langle \alpha \rangle \beta$ for each $\beta \in f_{\delta}(\alpha)$,
that $\mathfrak{M} \models [\alpha] \bigvee_{\beta \in f_{\delta}(\alpha)} \beta$,
and that $\mathfrak{M} \models [\alpha]\bot$ whenever $f_\delta(\alpha)=\;\varnothing$. Then $\mathfrak{M} \models \delta$.
\label{norm-sf-sat}
\end{lemma}
\begin{proof}
We show that $\mathfrak{M}$ satisfies every conjunct in (\ref{normal}); satisfaction of $\pi$ is given.
For conjuncts $\langle \alpha_k \rangle \beta_k$, for $k \in K$, suppose first
that $\alpha_k \neq \alpha_i$, for any $i \in I$.
Then $f_{\delta}(\alpha_k) = \{ \alpha_{k} \land \beta_{k'} : k' \in K$ such that $\alpha_k = \alpha_{k'} \}$.
If $\mathfrak{M} \models \langle \alpha_k \rangle (\alpha_{k} \land \beta_{k'})$
then
$\mathfrak{M} \models \langle \alpha_{k'} \rangle \beta_{k'}$ for all such $k'$.
Thus suppose
$\alpha_k = \alpha_i$ for some $i \in I$. Again by the construction of
$f_{\delta}$,
we have
$f_{\delta}(\alpha_k) = \{ \alpha_k \land \beta_{k'} \land \beta_{j_{k'}} : k' \in K$ such that $\alpha_k = \alpha_{k'}\}$
for some $j_{k'}$ where each $j_{k'} \in J_i$.
Then $\mathfrak{M} \models \langle \alpha_k\rangle (\alpha_k \land \beta_{k'} \land \beta_{j_{k'}})$
implies $\mathfrak{M} \models \langle \alpha_k\rangle \beta_{k'}$ for each such $k'$.
To see that $\mathfrak{M} \models [\alpha_i] \bigvee_{j \in J_i} \beta_j$
for each $i$ such that $\alpha_i = \alpha_k$, by the assumption, we have $\mathfrak{M} \models [\alpha_k] \bigvee_{j \in J_i'} (\alpha_k \land \beta_{k'} \land \beta_j)$ for some $J_i' \subseteq J_i$. Generalizing the disjunction
to $J_i$ and distributing it through shows that $\mathfrak{M} \models [\alpha_k] \bigvee_{j \in J_i} \beta_j$.
Finally, for
conjuncts $[\alpha_i] \bigvee_{j \in J_i} \beta_j$ where
$\alpha_i \neq \alpha_k$ for any $k \in K$, we have $f_{\delta}(\alpha_i) = \varnothing$
so that $\mathfrak{M} \models [\alpha_i] \bot$. But then $\mathfrak{M} \models [\alpha_i] \beta$ for any $\beta$ whatsoever, so that such conjuncts are satisfied.
\end{proof}

\begin{definition}
Let $\mathsf{PL}$ be a programming language
whose programs are the instances of \textit{$\langle$prog$\rangle$}
in the following grammar: 
\begin{grammar}
    <const> ::= `0' | `1'  \quad\quad\quad <var> ::= $X_1$ | $X_2$ | $\dots$ | $X_n$ | $\dots$

    <cond> ::= <var> `=' <const> | <var> `=' <var> \alt <var> `!=' <var> | <cond> `&' <cond>

    <assign> ::= <var> `:=' <const> | <var> `:=' <var> | <var> `:= !' <var>
    
    <branches> ::= <prog> | <branches> `or' <branches>

    <prog> ::= `' | <assign> | <prog> `;' <prog> | `loop' 
    \alt `if' <cond> `then' <prog> `else' <prog> `end'
    \alt `choose' <branches> `end'
\end{grammar}
$\mathsf{PL}_{\mathrm{det}}$ will denote the same language except that $\mathsf{PL}_{\mathrm{det}}$ excludes
\texttt{choose}-statements, $\mathsf{PL}^{\downarrow}$ is identical but for excluding \texttt{loop}-statements, and $\mathsf{PL}^{\downarrow}_{\mathrm{det}}$
is identical but for excluding both
\texttt{choose}-
and \texttt{loop}-statements.
\label{pl}
\end{definition}

A program in any of these languages may be ``compiled'' to the right type of Turing machine in an obvious way (\texttt{loop} represents an unconditional infinite loop). 
For the remainder of the article, fix $\mathsf{PL}^{\dagger}$ to be the programming language
of Defn. \ref{pl} corresponding to the choice of $\mathcal{M}^{\dagger}$.

With the normal form result
and suitable languages in hand,
we proceed to construct the canonical model $\mathfrak{M}_{\varphi} = (\mathsf{T}_{\varphi}, \mathbf{x}_{\varphi})$
for $\varphi$.
$\mathfrak{M}_{\varphi}$ need only satisfy a consistent clause
$\delta$ as in (\ref{normal}).
Intuitively, $\mathfrak{M}_{\varphi}$ will satisfy $\mathcal{L}_{\mathrm{prop}}$-atoms in $\delta$
via a suitable tape state $\mathbf{x}_{\varphi}$ (existent as $\delta$ and \emph{a fortiori}
$\pi$ is consistent), and will
satisfy each $\mathcal{L}_{\mathrm{cond}}$-atom by dint of
a branch in
$\mathsf{T}_{\varphi}$,
conditional on the antecedent, in which the consequent is made to hold.
We now write the
$\mathsf{PL}^{\dagger}$-code of such a $\mathsf{T}_{\varphi}$.

Suppose we are given $\delta$, and that for each $\alpha \in \mathcal{S}_{\delta}$
 we have code $\texttt{HoldsFromIntervention}(\alpha)$ defining
a condition that is met iff
the program is currently
being run under an intervention that fixes $\alpha$ to be true.
Then consider a $\mathsf{PL}$-program $\mathsf{P}_{\varphi}$
that contains one $\texttt{if}$-statement for each
$\alpha \in \mathcal{S}_{\delta}$, each executing if $\texttt{HoldsFromIntervention}(\alpha)$
is met. In the body of the $\texttt{if}$-statement
for $\alpha$,
$\mathsf{P}_{\varphi}$ has
a $\texttt{choose}$-statement
with one branch for each $\beta \in f_{\delta}(\alpha)$.
The branch for each $\beta$
consists of
a sequence of assignment 
statements guaranteed to make $\beta$ hold,
call this $\texttt{MakeHold}(\beta)$,
clearly existent since each $\beta$ is satisfiable. 
If $f_{\delta}(\alpha)$ is a singleton, this body contains only $\texttt{MakeHold}(\beta)$;
if $f_{\delta}(\alpha) = \varnothing$, then this body
consists of a single $\texttt{loop}$-statement.
If $\mathsf{T}_{\varphi}$ is the machine
corresponding to $\mathsf{P}_{\varphi}$,
and $\mathbf{x}_{\varphi}$ is a tape state
satisfying $\pi$,
then
$\mathfrak{M}_{\varphi} \models \langle \alpha \rangle \beta$ for each
$\beta \in f_{\delta}(\alpha)$,
as the program has a halting branch with $\texttt{MakeHold}(\beta)$;
also, $\mathfrak{M}_{\varphi} \models [ \alpha ] \bigvee_{\beta \in f_{\delta}(\alpha)} \beta$
as there are no other halting executions.
If $f_{\delta}(\alpha) = \varnothing$, then
$\mathfrak{M}_{\varphi} \models [\alpha]\bot$,
since under an $\alpha$-fixing intervention
the program
reaches a $\texttt{loop}$-statement
and has no halting executions.
So by Lem. \ref{norm-sf-sat},
we have that
$\mathfrak{M}_{\varphi}$
satisfies $\delta$. And thus $\varphi$.
To see that $\mathfrak{M}_{\varphi} \in \mathcal{M}^{\dagger}$, apply
Lem. \ref{sf-prop}:
in $\textsf{AX}^{\downarrow}$, $\varnothing \notin \mbox{range}(f_{\delta})$
so we have no $\texttt{loop}$s in $\mathsf{P}_{\varphi}$ and $\mathfrak{M}_{\varphi} \in \mathcal{M}^{\downarrow}$.
In $\textsf{AX}_{\mathrm{det}}$,
we have no $\texttt{choose}$-statements, so $\mathfrak{M}_{\varphi} \in \mathcal{M}_{\mathrm{det}}$; in $\textsf{AX}_{\mathrm{det}}^{\downarrow}$, we have neither $\texttt{loop}$- nor $\texttt{choose}$-statements,
and $\mathfrak{M}_{\varphi} \in \mathcal{M}_{\mathrm{det}}^{\downarrow}$.

But how do we know it is possible to write code
$\texttt{HoldsFromIntervention}(\alpha)$ by which the program can  tell
whether it is being run under
an $\alpha$-fixing intervention?
For any tape variable, we may try to toggle it.
If the attempt succeeds, then the variable is not presently fixed by an intervention.
If not, then the present execution is under an intervention fixing the variable.
Thus, we first try to toggle each \emph{relevant} variable.
Let $N$ be the maximum index
$i$ of any atom $X_i$ appearing in $\varphi$.
Listing \ref{isintervened}---call it $\texttt{IsIntervened}(X_i)$---
performs the toggle check for $X_i$ and records the result in $X_{i + N}$. It uses $X_{i+2N}$
as a temporary variable and ultimately leaves the value of $X_i$ unchanged.
\begin{lstlisting}[escapeinside={(*}{*)},caption={$\texttt{IsIntervened}(X_i)$},label={isintervened}]
(*$X_{i+N}$*) := (*$X_{i}$*);
(*$X_{i}$*) := ! (*$X_{i}$*);
(*$X_{i+2N}$*) := (*$X_{i}$*);
if (*$X_{i+2N}$*) = (*$X_{i+N}$*) then (*$X_{i+N}$*) := 1
else (*$X_{i+N}$*) := 0 end;
(*$X_{i}$*) := ! (*$X_{i}$*);
\end{lstlisting}
If
$\texttt{IsIntervened}(X_i)$ has already been run
for all $1 \le i \le N$,
$\texttt{HoldsFromIntervention}(\alpha)$
simply checks that
exactly those variables appearing in $\alpha$ have been marked as intervened on,
and that these have the correct values.
If
$\alpha$ is the $\mathcal{L}_{\mathrm{int}}$-equivalent of
$\lnot X_{i_1} \land \dots  \land      \lnot X_{i_k} \land  X_{i_{k+1}} \land \dots \land X_{i_n}$,
code for $\texttt{HoldsFromIntervention}(\alpha)$ is given in Listing \ref{holdsfromintervention}.
\begin{lstlisting}[escapeinside={(*}{*)},caption={$\texttt{HoldsFromIntervention}(\alpha)$},label={holdsfromintervention}]
(*$X_{i_1}$*) = 0 & (*$\dots$*) & (*$X_{i_k}$*) = 0 &
    (*$X_{i_{k+1}}$*) = 1 & (*$\dots$*) & (*$X_{i_n}$*) = 1 &
(*$X_{i_1 + N}$*) = 1 & (*$\dots$*) & (*$X_{i_k  + N}$*) = 1 &
    (*$X_{i_{k+1}  + N}$*) = 1 & (*$\dots$*) & (*$X_{i_n  + N}$*) = 1
\end{lstlisting}
Completing the description of the code of $\mathsf{P}_{\varphi}$ adumbrated earlier, $\mathsf{P}_{\varphi}$ consists of, in order:
\begin{enumerate}
\item One copy of $\texttt{IsIntervened}(X_i)$ for
each $1 \le i \le N$.
\item For each $\alpha \in \mathcal{S}_{\delta}$, an $\texttt{if}$-statement
with condition $\texttt{HoldsFromIntervention}(\alpha)$, whose body is:
\begin{enumerate}
\item a $\texttt{choose}$-statement with a branch for each $\beta \in f_{\delta}(\alpha)$, with body
$\texttt{MakeHold}(\beta)$, if $|f_{\delta}(\alpha)| \ge 2$;
\item a $\texttt{MakeHold}(\beta)$-snippet, if $|f_{\delta}(\alpha)|  = 1$;
\item or a single $\texttt{loop}$-statement if $f_{\delta}(\alpha) = \varnothing$.
\end{enumerate}
\end{enumerate}
\end{proof}

Note that $\mathsf{P}_{\varphi}$ never reads or writes  
a variable $X_i$ for $i > 3N$, and the relevant
$\mathsf{PL}^{\dagger}$-operations may be implemented with bounded space,
so that we have the following Corollary:
\begin{corollary}
\label{fsm-corollary}
Let $\mathcal{M}^{\dagger, \mathrm{fin}}$ be the class of
\emph{finite state machine} restrictions of $\mathcal{M}^{\dagger}$,
i.e. those $(\mathsf{T}, \mathbf{x}) \in \mathcal{M}^{\dagger}$
where $\mathsf{T}$ uses only boundedly many tape variables, for any input and intervention.
Then Thm. \ref{ax-soundcomplete} holds also for $\mathcal{M}^{\dagger, \mathrm{fin}}$.
\qed
\end{corollary} 

\section{Computational Complexity}
In this section we consider the
problem $\textsc{Sim-Sat}(\varphi)$ of deciding whether a given $\varphi \in \mathcal{L}$
is
satisfiable in $\mathcal{M}^{\dagger}$.
Although by Prop. \ref{model-checking},
it is in general
undecidable whether a given \emph{particular} simulation model satisfies a formula,
we show
here that it is decidable whether a given formula is satisfied by \emph{any} model. In fact, reasoning in this framework is no harder than reasoning in propositional logic:
\begin{theorem}
$\textsc{Sim-Sat}(\varphi)$ is $\mathsf{NP}$-complete in $|\varphi |$ (where $|\varphi|$ is defined standardly).
\label{complexity}
\end{theorem}

\begin{proof}
We clearly have $\mathsf{NP}$-hardness
as propositional satisfiability can be embedded directly into $\mathcal{L}$-satisfiability. 
To see that satisfiability is $\mathsf{NP}$, we guess
a $\mathfrak{M}$ and
check whether $\mathfrak{M} \models \varphi$.
$\mathcal{M}^{\dagger}$ is infinite,
and the checking step is undecidable by Prop. \ref{model-checking}.
So how could such an algorithm work?
The crucial insight is that
we may limit our search to a finite class $\mathcal{M}^{\dagger}_{\varphi}$
of models that are similar to
the canonical $\mathfrak{M}_{\varphi}$ (Lem. \ref{tdagger}).
Moreover,
a nice property of the canonical $\mathsf{T}_{\varphi}$ is that it wears its
causal structure on its sleeves:
one can read off the effect of any intervention
from the code of $\mathsf{P}_{\varphi}$,
and $\mathsf{P}_{\varphi}$ has polynomial size in $|\varphi|$ (implied by Lem. \ref{plddagger}).
Models in $\mathcal{M}^{\dagger}_{\varphi}$ will share
this property, guaranteeing that
the checking step can be done in polynomial time.
We will now make
$\mathcal{M}^{\dagger}_{\varphi}$ precise and prove these claims.
Let $\mathcal{S}_{\varphi} \subset \mathcal{L}_{\mathrm{int}}$ denote
the set of
$\mathcal{L}_{\mathrm{cond}}$-antecedents appearing in
$\varphi$.
For $C \in \mathbb{N}$, define
$\mathsf{PL}^{\dagger}_{\varphi, C} \subset
\mathsf{PL}^{\dagger}$ as the fragment of programs whose code consists of:
\begin{enumerate}
\item One copy of $\texttt{IsIntervened}(X_i)$ (Listing \ref{isintervened}), for each $1 \le i \le N$,
followed by
\item at most one copy of an $\texttt{if}$-statement with condition $\texttt{HoldsFromIntervention}(\alpha)$ (Listing \ref{holdsfromintervention}) 
for each $\alpha \in \mathcal{S}_{\varphi}$, whose body is one and only one of the following options, (a)--(c):
\begin{enumerate} 
\item a $\texttt{choose}$-statement with at most $C |\varphi|$ branches,
each of which has a body consisting of a single sequence of assignments, which may only be to variables $X_i$ for $1 \le i \le N$;
\item a single sequence of assignment statements, only to variables $X_i$ for $1 \le i \le N$;
\item a single $\texttt{loop}$-statement.
\end{enumerate}
However,
if $\mathsf{PL}^{\dagger} = \mathsf{PL}_{\mathrm{det}}$, (a) is not allowed;
if $\mathsf{PL}^{\dagger} = \mathsf{PL}^{\downarrow}$, (c) is not allowed;
and if $\mathsf{PL}^{\dagger} = \mathsf{PL}^{\downarrow}_{\mathrm{det}}$, neither
(a) nor (c) is allowed.
\end{enumerate}

\begin{lemma}
\label{plddagger}
The maximum length (defined standardly) of a program in
$\mathsf{PL}^{\dagger}_{\varphi, C}$ is polynomial in $|\varphi |$,
and there is a $C$ such that  for all $\varphi$, we have $\mathsf{P}_{\varphi} \in \mathsf{PL}^{\dagger}_{\varphi, C}$, assuming $\mathsf{P}_{\varphi}$ exists.
\end{lemma}
\begin{proof}
$N$ is $\mathcal{O}(|\varphi |)$, so part 1 of a program
is $\mathcal{O}(|\varphi |)$ in length.
There are at most $|\mathcal{S}_{\varphi}|$ $\texttt{if}$-statements in part 2; consider
the body of
each one.
In case (a) it has $\mathcal{O}(|\varphi|)$ branches, each of which
involves assignment to at most $N$ variables, and thus
has length $\mathcal{O}(|\varphi|^2)$. In case (b) its length is
$\mathcal{O}(|\varphi|)$; in case (c) its length is $\mathcal{O}(1)$.
Since $|\mathcal{S}_{\varphi}|$ is $\mathcal{O}(|\varphi|)$, the total length of part 2 is $\mathcal{O}(|\varphi|^3)$, so that both parts combined are $\mathcal{O}(|\varphi|^3)$.
To show the existence of $C$, it suffices to prove:
any $\texttt{choose}$-statement in the body of an $\texttt{if}$-statement in
$\mathsf{P}_{\varphi}$ has
$\mathcal{O}(|\varphi|)$ branches.
Now, the number of branches in the $\texttt{if}$-statement for $\alpha$ is $|f_{\delta}(\alpha)|$,
for some consistent $\delta$ as in (\ref{normal}).
But (\ref{normal}) is a clause of the disjunctive normal form of $\varphi$ and contains no more
$\mathcal{L}$-literals than does $\varphi$, which is of course $\mathcal{O}(|\varphi|)$. 
Since each element of $f_{\delta}(\alpha)$ arises from the selection of a literal
in (\ref{normal}), the number of branches is $\mathcal{O}(|\varphi|)$.
\end{proof}
Henceforth let $\mathsf{PL}^{\dagger}_{\varphi}$ denote $\mathsf{PL}^{\dagger}_{\varphi, C}$
for some $C$ guaranteed by Lem. \ref{plddagger}, and
call the set of $\mathbf{x}$ where only tape variables $X_i$ with indices $1 \le i \le N$
are possibly nonzero $\mathcal{X}_N$.
Let $\mathcal{M}^{\dagger}_{\varphi}$ be the class of
models $(\mathsf{T}, \mathbf{x})$ where $\mathsf{T}$ comes from a $\mathsf{PL}^{\dagger}_{\varphi}$-program and $\mathbf{x} \in \mathcal{X}_N$.
$\mathcal{M}^{\dagger}_{\varphi}$ is finite, and the following Lemma guarantees that we may restrict the search to $\mathcal{M}^{\dagger}_{\varphi}$:
\begin{lemma}
\label{tdagger}
$\varphi$ is satisfiable
with respect to $\mathcal{M}^{\dagger}$ iff it is satisfiable with respect to
$\mathcal{M}^{\dagger}_{\varphi}$. 
\end{lemma}
\begin{proof}
If $\varphi$ is satisfiable in $\mathcal{M}^{\dagger}$, it
is $\textsf{AX}^{\dagger}$-consistent by soundness, and hence has a canonical
$(\mathsf{T}_{\varphi}, \mathbf{x}_{\varphi})$. Without loss of generality take
$\mathbf{x}_{\varphi}$ from Thm.
\ref{ax-soundcomplete} to be in $\mathcal{X}_N$.
Then by Lem. \ref{plddagger}, $(\mathsf{T}_{\varphi}, \mathbf{x}_{\varphi}) \in \mathcal{M}^{\dagger}_{\varphi}$, so $\varphi$ is satisfiable in $\mathcal{M}^{\dagger}_{\varphi}$.
\end{proof}

Now with Lem. \ref{tdagger} our algorithm will guess a program $\mathsf{P} \in \mathsf{PL}^{\dagger}_{\varphi}$
and a tape $\mathbf{x} \in \mathcal{X}_N$,
and verify whether the guessed model $\mathfrak{M} \in \mathcal{M}^{\dagger}_{\varphi}$ satisfies $\varphi$.
We just need to show that the verification step is decidable in polynomial time.
Suppose that all negations in
$\varphi$ appear only before $\mathcal{L}$-atoms, since any formula may
be converted to such a form in linear time. Further, rewrite
literals of the form $\lnot [\alpha] \beta$ to 
$\langle \alpha\rangle \lnot \beta$. 
Then it suffices to show that we can decide in polynomial time
whether $\mathfrak{M}$
satisfies a given literal in $\varphi$:
there are linearly many of these and the truth-value of $\varphi$ may be evaluated
from their values in linear time.
For an $X$-literal $X_i$
or $\lnot X_i$, we simply output whether or not $\mathbf{x} \models X_i$.
For $\mathcal{L}_{\mathrm{cond}}$-literals with antecedent $\alpha$,
simulate execution of
$\mathcal{I}_{\alpha}(\mathsf{T})$ 
on $\mathbf{x}$.
Because $\mathsf{P} \in \mathsf{PL}^{\dagger}_{\varphi}$
and such programs
trigger at most one $\texttt{HoldsFromIntervention}(\alpha)$
$\texttt{if}$-statement when run
under an intervention, we may
perform this simulation by checking if there is any
$\texttt{if}$-statement for $\alpha$ in $\mathsf{P}$. If so, do one of the following, depending on 
what its body contains:
\begin{enumerate}[label=(\alph*)]
\item If a $\texttt{choose}$-statement, simulate the result of running
each branch. Output true iff: either the literal was $[\alpha] \beta$ and every
resulting tape satisfies
$\beta$, or the literal was $\langle \alpha \rangle \beta$ and at least one resulting
tape satisfies $\beta$.
\item If an assignment sequence, simulate running it on the current tape,
and output true iff the resulting tape satisfies $\beta$.
\item If a $\texttt{loop}$, output true iff the literal is of the $[\alpha] \beta$ form.
\end{enumerate}
This
algorithm is correct since we thereby capture all halting executions, given
that $\mathsf{PL}^{\dagger}_{\varphi}$-programs conform to the fixed structure above. That it
runs in polynomial time follows from the polynomial-length bound of Lem. \ref{plddagger}.
\end{proof}

\section{Conclusion and Future Work} 

A very natural way to assess a claim, `if $\alpha$ were true, then $\beta$ would be true,' is to run a simulation in which $\alpha$ is assumed to hold and determine whether $\beta$ would then follow. Simulations can be built using any number of tools: (probabilistic) programming languages designed specifically for generative models, generative neural networks, and many others. Our formulation of \emph{intervention on a simulation program} is intended to subsume all such applications where conditional reasoning seems especially useful. We have shown that this general way of interpreting conditional claims has its own distinctive, and quite weak, logic. Due to the generality of the approach, we can validate further familiar axioms by restricting attention to smaller classes of programs (deterministic, always-halting). We believe this work represents an important initial step in providing a foundation for conditional reasoning in these increasingly common contexts.  

To close, we would like to mention several notable future directions. Perhaps the most obvious next step is to extend our treatment to richer languages, and in particular to the first order setting. This is pressing for several reasons. First, much of the motivation for many of the generative frameworks mentioned earlier was to go beyond the propositional setting characteristic of traditional graphical models, for example, to be able to handle unknown (numbers of) objects (see \cite{Poole,BLOG}). 

Second, much of the work in conditional logic in AI has dealt adequately with the first order setting by using frameworks based on normality orderings \cite{Delgrande,FriedmanKollerHalpern}. It is perhaps a strike against the structural equation approach that no one has shown how to extend it adequately to first order languages (though see \cite{Halpern2000} for partial suggestions). In the present setting, just as we have used a tape data structure to encode a propositional valuation, we could also use such data structures to encode first order models. The difficult question then becomes how to understand complex (i.e., arbitrary first-order) interventions. We have begun exploring this important extension.

Given the centrality of probabilistic reasoning for many of the aforementioned tools,
it is important to consider the probabilistic setting.
Adding explicit probability operators in the style of \cite{Fagin} results in a very natural
extension of the system \cite{Ib18}.
One could also use probability thresholds (see, e.g, \cite{Hawthorne}):
we might say $(\mathsf{T},\mathbf{x}) \models [\alpha]\beta$ just when
$\mathcal{I}_{\alpha}(\mathsf{T})$ results in output satisfying $\beta$ with at least some threshold
probability.

Finally, another direction is to consider additional subclasses of programs, even for the basic propositional setting we have studied here. For example, in some contexts it makes sense to assume that variables are time-indexed and that no variable depends on any variable at a later point in time (as in dynamic Bayesian networks \cite{Dean}). In this setting there are no cyclic dependencies, which means we do not have programs like that in Example \ref{ex1}. Understanding the logic of such classes would be worthwhile, especially for further comparison with central classes of structural equation models (such as the ``recursive'' models of \cite{Pearl2009}).

\section*{Acknowledgments}

Duligur Ibeling is supported by the Sudhakar and Sumithra Ravi Family Graduate Fellowship in the School of Engineering at Stanford University.
\bibliographystyle{named}
\bibliography{ijcai18}

\end{document}